\documentclass
[aps,preprint,preprintnumbers,amsmath,amssymb,nofootinbib]{revtex4}%
\usepackage{amssymb}
\usepackage{amsmath}
\usepackage{bm}
\usepackage{amsthm}
\usepackage[a4paper,bindingoffset=0.2in,
left=20mm,right=20mm,top=30mm,bottom=30mm,            footskip=.25in]%
{geometry}
\usepackage{tikz}
\usepackage{lipsum}
\usepackage{url}
\usepackage{setspace}
\usepackage{afterpage}
\usepackage[titletoc]{appendix}
\usepackage{amsfonts}
\usepackage{graphicx}%
\providecommand{\U}[1]{\protect\rule{.1in}{.1in}}
\usetikzlibrary{matrix}

\theoremstyle{plain}
\newtheorem{thm}{Theorem}
\theoremstyle{definition}

\theoremstyle{proposition}

\theoremstyle{lemma}
\newtheorem{lemma}{Lemma}
\theoremstyle{corollary}

\begin{document}
\title{On the Cauchy Problem for Weyl-Geometric Scalar-Tensor Theories of Gravity}
\author{R. Avalos$^{1}$, I. P. Lobo$^{2}$, T. Sanomiya$^{2}$ and C. Romero$^{2}$\\
${}^{2}\!$Departamento de F\'{\i}sica, Universidade Federal da Para\'{\i}ba, Caixa
Postal 5008, 58059-970 Jo\~{a}o Pessoa, PB, Brazil.\\
${}^{1}\!$Departamento de Matem\'atica - UFC, Bloco 914 – Campus do Pici, 60455-760 Fortaleza, Cear\'a,  Brazil.\\
E-mail: rodrigo.avalos@fisica.ufpb.br; iarley\_lobo@fisica.ufpb.br; sanomiya@fisica.ufpb.br; cromero@fisica.ufpb.br}

\begin{abstract}
In this paper, we analyse the well-posedness of the initial value formulation
for particular kinds of geometric scalar-tensor theories of gravity, which
are based on a Weyl integrable space-time. We will show that, within a
frame-invariant interpretation for the theory, the Cauchy problem in vacuum is
well-posed. We will analyse the global in space problem, and, furthermore, we
will show that geometric uniqueness holds for the solutions. We make contact
with Brans-Dicke theory, and by analysing the similarities with such models,
we highlight how some of our results can be translated to this well-known
context, where not all of these problems have been previously addressed.

\end{abstract}
\maketitle

\section{Introduction}

In this paper, we will analyse the initial value formulation for special types
of geometric scalar-tensor theories, which are based on a Weyl integrable
structure as a model for space-time. This type of structures have bee proposed
as a suitable model for space-time by several authors \cite{Romero1}%
,\cite{Salim-Poulis},\cite{Romero2},\cite{CBPF},\citep{Scholz1}. In fact, it
is quite interesting that Weyl manifolds appear quite naturally as suitable
models for space-time within the well-known axiomatic approach to space-time
put forward by Elhers, Pirani and Schild (EPS) in \cite{EPS}. There, it is
shown that starting from basic assumptions about the behaviour of light and
freely falling particles, space-time naturally acquires the structure of a
general Weyl manifold. Then, it is shown that by adding some hypothesis about
the behaviour of clocks, the Weyl structure should be integrable,
\textit{i.e}, Weyl's 1-form must be exact. A slightly different approach is
followed in \cite{ADR}, where a rigorous discussion of the second-clock effect
in the context of general Weyl structures is used to show that, without
invoking any additional axioms, the rejection of such an effect as a realistic
physical phenomenon, reduces the scenario to a Weyl integrable structure.

The ideas presented above, show that Weyl structures appear naturally as
interesting models for space-time. In this context such structures have been
explored by some cosmologists and, in particular, in \cite{Romero1} it has
been pointed out that Weyl geometry might provide a natural geometric
framework for some scalar-tensor theories. In such a context, the scalar field
would appear as part of the geometric structure of space-time. For instance,
in \cite{Romero1} it has been shown that applying a Palatini-type variational
principle to the Brans-Dicke action, Weyl's compatibility condition naturally
appears as the link between the metric, the scalar field and the connection.
Thus, it has been proposed that such action should be posed on a Weyl
integrable structure, and that in order to make the theory frame independent,
it is argued that only geometrical quantities invariant under Weyl
transformations should be regarded as containing physical information. This
view may be supported by the fact that, following EPS, freely falling
particles are shown to follow geodesics of the Weyl structure, which are
invariant under Weyl transformations. Thus, each frame lies in equal footing
when describing the motion of such particles. This point of view has been
adopted in the study of some cosmological models \cite{Romero3}%
,\citep{Romero4}. Also, a recent review on the use of such Weyl structures in
modern physics can be found in \cite{Scholz2}.

Another important issue that has been raised, related with the above
paragraph, is that the kind of model proposed in \cite{Romero1} is intimately
related with Brans-Dicke theory. In fact, using different representatives in
the equivalence class of Weyl integrable manifolds, the field equations take
the same form as in the Brans-Dicke theory in either the Jordan frame or the
Einstein frame, but, within the invariant interpretation presented in
\cite{Romero1}, both frames would be equivalent. (For a discussion on this matter, see \cite{Iarley}).

With the above considerations in mind, we will analyse whether the initial
value formulation for the the theory proposed in \cite{Romero1} is well-posed,
within the invariant formulation presented therein. Since the initial value
formulation of any physical theory is one of the things which lie in the core
of its interpretation, giving it predictability, we regard this issue as
fundamental if the theory is to be taken as \textit{realistic}.

This paper will be organized as follows. We will first review the type of
models we will be analysing and remind the reader the basic definitions
concerning Weyl manifolds. Then, we will address the initial value formulation
of such theories, and, finally, we will conclude with some remarks and future perspectives.

\section{Weyl-geometric scalar-tensor theories (WGSTT)}

\subsection*{Weyl geometry}

A Weyl manifold is a triple $(M,g,\varphi)$, where $M$ is a differentiable
manifold, $g$ a semi-Riemannian metric on $M$, and $\varphi$ a 1-form field on
$M$. We assume that $M$ is endowed with a torsion-free linear connection
$\nabla$ satisfying the following \textit{compatibility condition}:
\begin{align}
\label{compatibility}\nabla g=\varphi\otimes g.
\end{align}
where, $\nabla g$ denotes the $(0,3)$-tensor field defined by $\nabla
g(X,Y,Z)\doteq(\nabla_{X}g)(Y,Z)$, where $X,Y$ and $Z$ are vector fields.
Results concerning the existence and uniqueness of such a connection are
straightforward and their proofs are analogous to those known for the
Riemannian case \cite{Folland}. It can easily be shown that, in local
coordinates, the components of the Weyl connection $\nabla$ are given by:
\[
\Gamma_{ac}^{u}=\frac{1}{2}g^{bu}(\partial_{a}g_{bc}+\partial_{c}%
g_{ab}-\partial_{b}g_{ca})+\frac{1}{2}g^{bu}(\varphi_{b}g_{ca}-\varphi
_{a}g_{bc}-\varphi_{c}g_{ab}).
\]

A \textbf{Weyl integrable manifold} is defined to be a Weyl manifold where
Weyl's 1-form $\varphi$ is exact, that is, $\varphi=d\phi$ for some (smooth)
scalar function $\phi$ on $M$.

Let us now note that the compatibility condition (\ref{compatibility}) is
invariant under the following group of transformations:
\begin{equation}
\left\{
\begin{array}
[c]{ll}%
\overline{g}=e^{-f}g & \\
\overline{\varphi}=\varphi-df &
\end{array}
\right.  \label{weyl group}%
\end{equation}
where $f$ is an arbitrary smooth function defined on $M$. By this, we mean
that if $\nabla$ is compatible with $(M,g,\varphi),$ then it is also
compatible with $(M,\overline{g},\overline{\varphi})$. It is easy to check
that these transformations define an equivalence relation between Weyl
manifolds. In this way, every member of the class is compatible with the same
connection, hence has the same geodesics, curvature tensor and any other
property that depends only on the connection. This is the reason why it is
regarded more natural, when dealing with Weyl manifolds, to consider the whole
class of equivalence $(M,[g,\varphi])$ rather than working with a particular
element of this class. In this sense, it is argued that only geometrical
quantities that are invariant under (\ref{weyl group}) are of real
significance in the case of Weyl geometry. Following the same line of
argument, it is also assumed that only physical theories and physical
quantities presenting this kind of invariance should be considered of interest
in this context.

An interesting fact to note about integrable Weyl structures, is that within
each equivalence class $(M,[g,\phi])$, there is a representative for which its
Weyl compatible connection coincides with the Riemannian connection related to
this element. This is clear, since if we consider the element $(M,g^{\prime
},\phi^{\prime})$ given by $g^{\prime}=e^{-\phi}g$ and $\phi^{\prime}=0$, we see
that the above argument holds.

\subsection*{WGSTT}

In this section we will shortly review the main ideas about the models we will
be considering in the remainder of the paper. The starting point of the model
proposed in \cite{Romero1}, is an action functional of the following type
\begin{align}
\label{action1}S(\nabla,g,\phi)=\int_{V}e^{-\phi}(R(\nabla)+\omega\langle
d\phi,d\phi\rangle)\mu_{g},
\end{align}
where $V$ stands for a 4-dimensional space-time, $g$ for a Lorentzian metric,
$\phi$ for a scalar field, $\omega$ is a coupling constant and $\mu_{g}$ the
volume form associated with $g$. In this context the Ricci scalar $R$ is
associated to a torsionless connection $\nabla$.

It should be remarked that, if we consider that the background geometry is
Riemannian, then (\ref{action1}) is equivalent to the Brans-Dicke action. The
idea adopted in \cite{Romero1} is to consider a Palatini-type action principle
for (\ref{action1}), and thus determine the relation between the connection
and the other geometric objects (in this context $g$ and $\phi$) from the
field equations. Interestingly enough, this procedure leads to the Weyl
compatibility condition between the connection, the metric and the scalar
field. Thus, we are led to a theory based on a Weyl integrable space-time.
Written as in (\ref{action1}), this theory is posed in vacuum, but a
prescription for coupling other fields is given \cite{Romero1}.

It should be noted that, once it is established that the physical theory which
comes from (\ref{action1}) should be posed on a Weyl integrable manifold, it
is straightforward to try pose it on a Weyl integrable structure. This is a
consequence of the natural invariance of the connection under Weyl
transformations. Furthermore, following for instance the EPS axiomatic
approach, we conclude that such structures appear as the most general
mathematical models for space-time within the original scope proposed in
\cite{EPS}. In this scenario, freely falling point-like massive particles
follow geodesics of the Weyl connection. This is a consequence of taking the
equivalence principle as one of the axioms for space-time. Thus, the
trajectories of massive particles will be invariant under Weyl
transformations. These ideas point towards adopting the view that physical
properties, in this context, should be invariant under Weyl transformations,
and thus that the theory described by (\ref{action1}) should be posed on the
equivalence class $(V,[g,\phi])$.

Taking into account the remarks above, we should stress that the second term in (\ref{action1}) is not invariant under Weyl transformations. Thus, we must expect the field equations derived from such action to be \textit{frame}-dependent. Therefore, from (\ref{action1}) we will get a set of frame-dependent field equations, whose solutions define equivalence classes of the form $(V,[g,\phi])$, and, from the physical point of view, any member in the equivalence class is on equal footing, and thus any physical property described by such a solution must be invariant under Weyl transformations. This point of view has been used to describe some cosmological scenarios in \cite{Romero3},\cite{Romero4}.

Finally, we should stress that, as might be expected, there is a strong
connection between the theory described above and Brans-Dicke theory. In fact,
another interesting aspect of this geometric model, is that the frame
transformations connected with the so called Einstein and Jordan frames in
Brans-Dicke theory, can be understood as Weyl transformation within this
context \cite{Iarley}.

The issue of the equivalence of frames in Brans-Dicke theory has been treated
(from the classical point of view) in some previous works
\cite{Faraoni:2006fx,Quiros:2011wb}. In particular, \cite{Faraoni:2006fx}
recaptures the arguments of Dicke's original paper \cite{Dicke:1961gz}, in
which the notion of running units is addressed and applied to analyse the
equivalence between frames. In this line of argument, it is claimed that any
measurement of a physical observable is performed by considering the ratio
between such a quantity, derived in a certain frame, and some chosen unit.
Being that this ratio is what is actually observed, it is this quantity which should
be invariant under frame transformations. Therefore, although the field
equations are not invariant, the measured physical observable would be.

Contrasting with the above interpretation, the kind of models proposed in
\cite{Romero1} also consider that different frames should be treated on equal
footing, but this idea is translated into the mathematical structure on which
the models are based. This is accomplished by considering space-time, and
physical phenomena, described by an equivalence class $(V,[g,\phi])$ defining
a Weyl integrable structure, and proposing that physical observables should be
invariant under Weyl transformations. It is worth noticing that a simple way
of defining invariant quantities in this formulation, consists in using the
invariant metric $\gamma=e^{-\phi}g$ to perform contractions on invariant tensors.


To conclude this section, we will write down the $d$-dimensional version of
the action functional (\ref{action1}) and the field equations we get from it.
We will pose the problem on a Weyl integrable manifold from the beginning.
Thus, the curvature terms are to be understood as computed with the Weyl
connection. For quantities computed with a Riemannian connection we will take
as a notational convention to put the symbol `${}^{\circ}$' on the upper left
corner of such quantities.

Our $d$-dimensional action is the following:
\begin{align}
\label{action}S(g,\phi)=\int_{V}e^{-\frac{d-2}{2}\phi}(R+\omega\langle
d\phi,d\phi\rangle-V(\phi))\mu_{g},
\end{align}
where, in contrast to (\ref{action1}), we have included a potential for the
scalar field in order to take into account some cosmological scenarios, such
as the ones discussed in \cite{Romero4}. Straightforward computations show
that the Euler-Lagrange equations for such an action are the following:
\begin{align}
\label{fieldeqs1}%
\begin{split}
&  G_{\alpha\beta}=-\omega\nabla_{\alpha}\phi\nabla_{\beta}\phi+\frac{1}%
{2}g_{\alpha\beta}(\omega\nabla_{\sigma}\phi\nabla^{\sigma}\phi-V(\phi))\\
&  2\omega\nabla_{\sigma}\nabla^{\sigma}\phi+\frac{d+2}{2}\omega\nabla
_{\sigma}\phi\nabla^{\sigma}\phi+\frac{d-2}{2}R+V^{\prime}(\phi)-\frac{d-2}%
{2}V(\phi)=0.
\end{split}
\end{align}

These equations can be rewritten in a more convenient way by taking the trace
of the first equations and using this information in the second one. After a
few computations we get the following:
\begin{align}
\label{fieldeqs2}%
\begin{split}
&  G_{\alpha\beta}=-\omega\nabla_{\alpha}\phi\nabla_{\beta}\phi+\frac{1}%
{2}g_{\alpha\beta}(\omega\nabla_{\sigma}\phi\nabla^{\sigma}\phi-V(\phi)),\\
&  g^{\sigma\mu}\nabla_{\sigma}\nabla_{\mu}\phi+\frac{1}{2\omega}(V^{\prime
}(\phi)+V(\phi))=0.
\end{split}
\end{align}

Recall that the point of view that we are taking is that a solution of these
field equations defines an equivalence class $(V,[g,\phi])$, where each
element is related to another by a Weyl transformation, and that there is no
preferred element in the class to describe physical properties. Thus, it will
be instructive to look at the field equations satisfied by a generic element
in the class $(V,[g,\phi])$ defined by a solution of (\ref{fieldeqs2}). Since
any element $(V,g^{\prime},\phi^{\prime})$ in the equivalence class of
$(V,g,\phi)$ is of the form $(V,e^{-f}g,\phi-f)$ for some $f\in C^{\infty}%
(V)$, then $(g,\phi)=(e^{f}g^{\prime},\phi^{\prime}+f)$. Replacing these
expressions in (\ref{fieldeqs2}), we see that $(g^{\prime},\phi^{\prime})$
satisfy the following system:%

\begin{align}
\label{transformedsystem}%
\begin{split}
&  G_{\alpha\beta}(g^{\prime},\phi^{\prime})=-\omega\nabla_{\alpha}%
(\phi^{\prime}+f)\nabla_{\beta}(\phi^{\prime}+f)+\frac{\omega}{2}g^{\prime
}_{\alpha\beta}g^{\prime\mu\nu}\nabla_{\mu}(\phi^{\prime}+f)\nabla_{\nu}%
(\phi^{\prime}+f)-\frac{1}{2}e^{f}g^{\prime}_{\alpha\beta}V(\phi^{\prime
}+f),\\
&  g^{\prime\mu\nu}\nabla_{\mu}\nabla_{\nu}(\phi^{\prime}+f)+\frac{e^{f}%
}{2\omega}(V^{\prime}(\phi^{\prime}+f)+V(\phi^{\prime}+f))=0.
\end{split}
\end{align}
Where we have used the fact that the Einstein tensor is invariant under Weyl
transformations, thus $G_{\alpha\beta}(e^{f}g^{\prime},\phi^{\prime
}+f)=G_{\alpha\beta}(g^{\prime},\phi^{\prime})$ $\forall$ $f$. Also, note that
$V^{\prime}$ stand for the ordinary derivative of the function $V$, it does
not represent any transformation law for the potential induced by the Weyl transformations.

\section{The initial value formulation for WGSTT}

In this section we will analyse the initial value formulation for the theory
described in the previous section. As usual, we will consider space-time to be
globally hyperbolic. Thus, $V=M\times\mathbb{R}$, where $V$ is an
$(n+1)$-dimensional space-time and $M$ an $n$-dimensional hypersurface. On
this setting we have the standard $(n+1)$-splitting for space-time. By this we
mean that we can pick local adapted frames of the form
\begin{align}
\label{adaptedframe}%
\begin{split}
e_{0}  &  \doteq\partial_{t}-\beta,\\
e_{i}  &  \doteq\partial_{i}, \;\; i=1,\cdots,n,
\end{split}
\end{align}
where $\partial_{i}$ denote the coordinate vector fields associated to the
coordinates $x^{i}$ of any adapted coordinate system to the product structure
$(x,t)$. In this context, the vector field $\beta$ is the standard
\textit{shift vector}, which merely represents the tangent component of
$\partial_{t}$ to each hypersurface $M\times\{t\}\cong M$. Thus, $e_{0}$ is,
by construction, orthogonal to each $M\times\{t\}$. It should be noted that
since all the metrics in an equivalence class $[g,\phi]$ are conformally
related, they produce the same orthogonal splitting on each tangent space, and
thus $\beta$ does not depend on the choice of representative in the class.

It is straightforward to see that the dual frame to (\ref{adaptedframe}) is
given by
\begin{align}
\label{adaptedcoframe}%
\begin{split}
\theta^{0}  &  = dt\\
\theta^{i}  &  = dx^{i}-\beta^{i}dt, \;\; i=1,\cdots,n.
\end{split}
\end{align}
Using these adapted frames, we see that we can write any space-time metric in
the conformal class in the following way:
\begin{align}
g=-N^{2}\theta^{0}\otimes\theta^{0} + g_{ij}\theta^{i}\otimes\theta^{j},
\end{align}
where the function $N>0$ is referred to as the \textit{lapse function}, which
is frame-dependent. From now on, when we write coordinate expressions, unless
explicitly stated, we will be referring to coordinates with respect to adapted
frames of the type described above.

\bigskip In what follows, our strategy will be to first analyse the
well-posedness of the initial value formulation in a frame, that is, to
analyse the well-posednees of the set of equations (\ref{fieldeqs2}), and then
to show that this is enough to guarantee the well-posedness for the class
$(V,[g,\phi])$. We will now be more specific on what we mean by saying that
the Cauchy problem is well-posed on the equivalence class.

Since we are regarding that every meaningful quantity, which can in fact be
observed, must be invariant under Weyl transformation, then, certainly, the
initial value formulation should satisfy this condition if this interpretation
is to hold. By this we mean that if we have an initial data set for the
evolution problem, which has a development into a space-time $(V,g,\phi)$,
then any other equivalent space-time $(V,g^{\prime},\phi^{\prime})$ related to
$(V,g,\phi)$ by a Weyl transformation, must also be the product of the
evolution of initial data on $M$, and, in a specific sense, both initial data
sets should be equivalent. Similarly, given two \textit{equivalent} initial
data sets, if one of them has a development into a space-time satisfying
(\ref{fieldeqs2}), then the other must also have a development into an
equivalent space-time, thus satisfying (\ref{transformedsystem}) for some
gauge function $f$. Pictorially, what we intend to show is that the following
commutative diagram holds.

\begin{figure}[th]
\centering
\includegraphics[width=100mm]{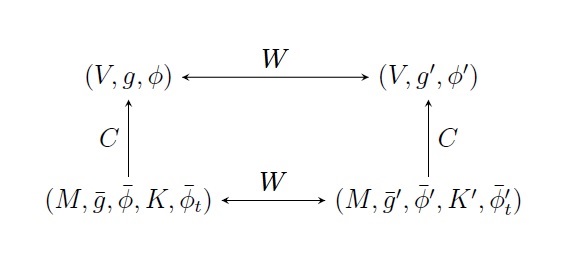}
\caption{
Commutativity of performing a Weyl transformation on the initial data sets and their respective space-time evolutions}
\end{figure}

In the above diagram $C$ stands for \textit{Cauchy evolution} and $W$ for Weyl
transformation; the quantities $\bar{g},\bar{\phi}$ define a Weyl structure on
the initial hypersurface; $K$ stands for a symmetric $(0,2)$-tensor field
that, after getting the embedding, represents the extrinsic curvature of the
embedded hypersurface $M\cong M\times\{0\}$, and $\bar{\phi}_{t}$ represents
the initial data for $\partial_{t}\phi(\cdot,0)$.

We should now make sense of what we mean by a Weyl transformation on the
initial data set. This is because, even though it is clear how a Weyl
transformation has to act on $\bar{g}$ and $\bar{\phi}$, the objects $K$ and
$\bar{\phi}_{t}$ are extrinsic objects, and the way they should transform
comes from looking at the embedded manifold $M\hookrightarrow M\times
\mathbb{R}$ and analysing their transformation behaviour. Thus, consider a
space-time Weyl transformation:
\begin{align}\label{spweyltrans}
g^{\prime}=e^{-f}g,\\
\phi^{\prime}=\phi-f.
\end{align}
Using the definition of the extrinsic curvature \cite{Oneill}, we have that
\begin{align*}
K^{\prime}(X,Y)\overset{def}{=}g^{\prime}(\nabla_{X}Y, n^{\prime}),
\end{align*}
where $n^{\prime}$ stands for the future-pointing unit normal vector field, in
the metric $g^{\prime}$, to each slice $M\times\{t\}$. It is clear that using
the adapted frames, we have that ${N^{\prime}}^{2}=e^{-f}N^{2}$, where $N$ and
$N^{\prime}$ are the lapse functions in the frames $(V,g,\phi)$ and
$(V,g^{\prime},\phi^{\prime})$. Then, we have that, by definition of the
lapse,
\begin{align*}
n^{\prime}=\frac{1}{N^{\prime}}e_{0}=e^{\frac{f}{2}}n.
\end{align*}
Then, clearly, we have that
\begin{align}\label{extcurvtrans}
K^{\prime}=e^{-\frac{f}{2}}K(X,Y)
\end{align}
That is, under a space-time Weyl transformation, the extrinsic curvature of
each hypersurface $M\times\{t\}$ transforms as $K^{\prime}=e^{-\frac{f}{2}}K$.
Thus, on the initial slice $M\cong M\times\{0\}$, setting $h(\cdot)\doteq
f(\cdot,0)\in C^{\infty}(M)$, we have the following induced Weyl
transformation on $M$
\begin{align*}
K^{\prime}|_{t=0}=e^{-\frac{h}{2}}K|_{t=0}.
\end{align*}
Also, with the same line of argument, we have that $\partial_{t}\phi^{\prime
}(\cdot,0)=\partial_{t}\phi(\cdot,0)-\partial_{t}f(\cdot,0)$. Thus, since
$\eta(\cdot)\doteq\partial_{t}f(\cdot,0)$ represents an arbitrary smooth
function on $M$, we get the following transformations on the initial data set
\begin{align}
\label{constrainttransf}%
\begin{split}
\bar{g}^{\prime}=e^{-h}\bar{g}, \;\;\;\;  &  ; \;\;\;\; \bar{\phi}^{\prime}%
=\bar{\phi}-h,\\
\bar{K}^{\prime}=e^{-\frac{h}{2}}\bar{K} \;\;\;\;  &  ; \;\;\;\; \bar{\phi}%
_{t}=\bar{\phi}_{t}-\eta.
\end{split}
\end{align}
Thus, we will take as a definition that a Weyl transformation on the initial
data set $(\bar{g},\bar{\phi},\bar{K},\bar{\phi}_{t})$ is given by
(\ref{constrainttransf}), where $h$ and $\eta$ are arbitrary smooth functions
on $M$.

\bigskip

It will be important to remark that, from the way $K$ transforms under a
space-time Weyl transformation, we can easily compute its coordinate
expression in an adapted frame. In order to do this, we simply use the fact
that for the Riemannian element $(V,e^{-\phi}g,0)$ in an equivalence class
$(V,[g,\phi])$, we know what the coordinate expression of $K^{\prime}$ looks
like, since this is a standard object in Riemannian geometry. Explicitly,
\begin{align*}
K^{\prime}_{ij}=-\frac{1}{2N^{\prime}}(\partial_{t}g^{\prime}_{ij}-(\nabla
_{i}\beta_{j}+\nabla_{j}\beta_{i})).
\end{align*}
Since we have found that $K(g,\phi)=e^{\frac{\phi}{2}}K^{\prime}(g^{\prime
},0)$, and we know that $N^{\prime-\frac{\phi}{2}}N$, we get that
\begin{align*}
K_{ij}(g,\phi)  &  =-\frac{e^{\phi}}{2N}(e^{-\phi}\partial_{t}g_{ij}-e^{-\phi
}g_{ij}\partial_{t}\phi-(\nabla_{i}\beta_{j}+\nabla_{j}\beta_{i})),\\
&  =-\frac{1}{2N}(\partial_{t}g_{ij}-g_{ij}\partial_{t}\phi-e^{\phi}%
(\nabla_{i}\beta_{j}+\nabla_{j}\beta_{i})).
\end{align*}
Thus, on the initial slice $M\cong M\times\{0\}$, we get
\begin{align}
\label{extcurv}\bar{K}_{ij}(g,\phi)=-\frac{1}{2\bar{N}}(\partial_{t}%
g_{ij}|_{t=0}-\bar{g}_{ij}\bar{\phi}_{t}-e^{\bar{\phi}}(\bar{\nabla}_{i}%
\bar{\beta}_{j}+\bar{\nabla}_{j}\bar{\beta}_{i})),
\end{align}
where all the quantities with a ``bar" on top are taken on $t=0$.

\subsection*{Constraint equations}

Just as in the case of general relativity (GR), the set of equations (\ref{fieldeqs2}) impose a set
of constraint equations for the initial data. Before writing these equations,
we will first stablish some notation. The field equation can be rewritten as
\begin{align}
&  G_{\alpha\beta}=T_{\alpha\beta},\\
&  g^{\alpha\beta}\nabla_{\alpha}\nabla_{\beta}\phi+\frac{1}{2\omega
}(V^{\prime}(\phi)+V(\phi))=0,
\end{align}
where
\begin{align}
\label{energy-momentum}T_{\alpha\beta}\doteq-\omega\nabla_{\alpha}\phi
\nabla_{\beta}\phi+\frac{1}{2}g_{\alpha\beta}(\omega\nabla_{\sigma}\phi
\nabla^{\sigma}\phi-V(\phi)).
\end{align}

The origin of the constraint equations is quite simple: If the above equations
equations are to hold on space-time, they must, in particular, hold on the
initial hypersurface $M\cong M\times\{0\}$. Thus, the equations
\begin{align}
\label{constraints6.0}%
\begin{split}
G(e_{0},e_{0})  &  =T(e_{0},e_{0}),\\
G(e_{0},e_{i})  &  =T(e_{0},e_{i}).
\end{split}
\end{align}
must be satisfied on the initial space-like slice. We will see that these
equation depend only on the initial data set. This implies that the initial
data set cannot be chosen arbitrarily. Only those initial data sets satisfying
the constraint equations (\ref{constraints6.0}) can have a development into a
space-time satisfying the above field equations.

We will now explicitly write down the constraint equations. In order to
compute the components of the Einstein tensor, we can take advantage of the
fact that we know their form in the Riemannian context, and that within each
class $(V,[g,\phi])$ there is a Riemannian element $(V,e^{-\phi}g,0)$,
\textit{i.e.}, for this element the Weyl connection coincides with its
Riemannian connection. Thus, since we know that the Einstein tensor is
invariant in the equivalence class, then, using the usual Gauss-Codazzi
equations from Riemannian geometry, we see that for
\begin{align*}
g^{\prime}=e^{-\phi}g,\\
K^{\prime}=e^{-\frac{\phi}{2}}K,
\end{align*}
the following relations hold:
\begin{align*}
G(e_{0},e_{0})  &  = \frac{{N^{\prime}}^{2}}{2}(({K^{\prime}}^{l}_{l}%
)^{2}-K^{\prime ij}K^{\prime}_{ij}+R(\bar{g}^{\prime})),\\
G(e_{0},e_{i})  &  =N^{\prime}(\bar{\nabla}_{i}{K^{\prime}}^{l}_{l}%
-\bar{\nabla}_{j}K^{\prime j}_{i}),
\end{align*}
where $\bar{\nabla}$ denotes the induced Weyl connection on $M$, compatible
with the induced Weyl structure $(M,[\bar{g},\bar{\phi}])$. Using
(\ref{constrainttransf}), we get
\begin{align}%
\begin{split}
G(e_{0},e_{0})  &  = \frac{{N}^{2}}{2}((K^{l}_{l})^{2}-K^{ij}K_{ij}+R(\bar
{g},\bar{\phi})),\\
G(e_{0},e_{i})  &  =N\bar{g}^{uj}(\bar{\nabla}_{i}K_{uj}-\bar{\nabla}%
_{j}K_{ui})+\frac{N}{2}(\bar{\nabla}_{j}\phi K^{j}_{i}-\bar{\nabla}_{i}\phi
K^{l}_{l}).
\end{split}
\end{align}

Now, in order to write down the constraints, we need to compute the same
components for the tensor $T$. Notice that, using an adapted frame, we have
\begin{align*}
T_{g,\phi}=\{-\omega e_{\alpha}(\phi)e_{\beta}(\phi)+\frac{1}{2}\omega
g_{\alpha\beta}g^{\mu\nu}e_{\mu}(\phi)e_{\nu}(\phi)-\frac{1}{2}g_{\alpha\beta
}V(\phi)\}\theta^{\alpha}\otimes\theta^{\beta}.
\end{align*}
From this expression it is straightforward to see that
\begin{align}%
\begin{split}
T(e_{0},e_{0})  &  =\frac{N^{2}}{2}\{-\omega(\pi^{2}+|\bar{\nabla}\phi
|^{2}_{\bar{g}})+V(\phi)\}\\
T(e_{0},e_{i})  &  =-\omega N\pi\bar{\nabla}_{i}\phi,
\end{split}
\end{align}
where $\pi\doteq\frac{1}{N}e_{0}(\phi)$, $\bar{g}$ denotes the induced metric
on each slice $M\times\{t\}$ and $\bar{\nabla}$ denotes the induced Weyl
connection. From these computations we get that the constraint equations are
the following
\begin{align}
\label{wistconstraints}%
\begin{split}
(K^{l}_{l})^{2}-K^{ij}K_{ij}+R  &  =-\omega\{\pi^{2}+|\bar{\nabla}\phi
|^{2}_{\bar{g}}\}+V(\phi),\\
\bar{g}^{uj}(\bar{\nabla}_{i}K_{uj}-\bar{\nabla}_{j}K_{ui})+\frac{1}{2}%
(\bar{\nabla}_{j}\phi K^{j}_{i}-\bar{\nabla}_{i}\phi K^{l}_{l})  &
=-\omega\pi\bar{\nabla}_{i}\phi.
\end{split}
\end{align}
The above expressions, on $M\cong M\times\{0\}$, depend only on the initial
data set, as previously remarked. Since these constraints are to be posed for
the initial data on the initial hypersurface $M$, it becomes clear that, just
as discussed for $K$, the way that $\pi$ transforms under a Weyl
transformation on the initial data set, does not come from the intrinsic Weyl
transformation on $M$, but is induced from the way it transforms under
space-time Weyl transformations. Also note that $\pi$ is a quantity that we
get by analysing the right-hand side of (\ref{fieldeqs2}), which is
frame-dependent. Thus, we should expect $\pi$ to be frame-dependent. In fact,
the above computations applied to the transformed system
(\ref{transformedsystem}), give us the following constraints, valid in each
frame:
\begin{align}%
\begin{split}
&  ({K^{*}}^{l}_{l})^{2}-{K^{*}}^{ij}{K^{*}}_{ij}+R(\bar{g}^{*},\bar{\phi}%
^{*})=-\omega\Big\{\big(\frac{e^{\frac{f(\cdot,0)}{2}}}{N}e_{0}(\bar{\phi}%
^{*}+f)\big)|_{t=0}^{2}+|\bar{\nabla} (\bar{\phi}^{*}+f(\cdot,0))|^{2}%
_{\bar{g}^{*}}\Big\}\\
&
\;\;\;\;\;\;\;\;\;\;\;\;\;\;\;\;\;\;\;\;\;\;\;\;\;\;\;\;\;\;\;\;\;\;\;\;\;\;\;\;\;\;\;\;\;\;\;\;\;\;\;\;\;\;\;\;\;\;\;\;\;\;\;\;\;\;\;\;\;\;\;\;\;\;\;\;\;\;\;\;\;\;\;\;\;\;\;\;\;\;\;\;\;\;\;\;\;\;\;\;\;\;\;\;\;
+e^{f(\cdot,0)}V(\phi^{*}+f)|_{t=0},\\
&  {\bar{g^{*}}}^{uj}(\bar{\nabla}_{i}{K^{*}}_{uj}-\bar{\nabla}_{j}{K^{*}%
}_{ui})+\frac{1}{2}(\bar{\nabla}_{j}\bar{\phi}^{*} {K^{*}}^{j}_{i}-\bar
{\nabla}_{i}{\bar{\phi}^{*}} {K^{*}}^{l}_{l})=-\omega\big(\frac{e^{\frac
{f(\cdot,0)}{2}}}{N}e_{0}(\bar{\phi}^{*}+f)\big)|_{t=0}\bar{\nabla}_{i}%
(\bar{\phi}^{*}+f(\cdot,0))\big).
\end{split}
\end{align}
where the notation comes from considering the space-time Weyl transformation
\begin{align*}
g^{*}  &  =e^{-f}g,\\
\phi^{*}  &  =\phi-f.
\end{align*}
From the above expression, we finally get that in an arbitrary frame,
\begin{align*}
\pi^{*}  &  \doteq\frac{e^{\frac{f}{2}}}{N}e_{0}(\phi^{*}+f)=\frac{e^{\frac
{f}{2}}}{N}e_{0}(\phi-f+f)\\
&  =e^{\frac{f}{2}}\pi.
\end{align*}
Thus, we can write the constraints of the transformed system
(\ref{transformedsystem}) as
\begin{align}
\label{transformedconstraints}%
\begin{split}
&  ({K^{*}}^{l}_{l})^{2}-{K^{*}}^{ij}{K^{*}}_{ij}+R(\bar{g}^{*},\bar{\phi}%
^{*})=-\omega\Big\{{\pi^{*}}^{2}+|\bar{\nabla} (\bar{\phi}^{*}+f(\cdot
,0))|^{2}_{\bar{g}^{*}}\Big\}+e^{f(\cdot,0)}V(\phi^{*}+f)|_{t=0},\\
&  {\bar{g^{*}}}^{uj}(\bar{\nabla}_{i}{K^{*}}_{uj}-\bar{\nabla}_{j}{K^{*}%
}_{ui})+\frac{1}{2}(\bar{\nabla}_{j}\bar{\phi}^{*} {K^{*}}^{j}_{i}-\bar
{\nabla}_{i}{\bar{\phi}^{*}} {K^{*}}^{l}_{l})=-\omega\pi^{*}\bar{\nabla}%
_{i}(\bar{\phi}^{*}+f(\cdot,0)).
\end{split}
\end{align}
Now, just as instead of $\partial_{t}\bar{g}|_{t=0}$, we actually consider $K$
as part of the geometric initial data, we will consider $\pi$ instead of
$\bar{\phi}_{t}$ as part of the geometric initial data. In this way, when
given the initial data set $(\bar{g},\bar{\phi},K,\pi)$, we will have two
\textit{gauge} choices: one for the lapse-shift and another one for the values
of $f,\partial_{t}f|_{t=0}$. In this way we get that an initial data set for
the system (\ref{fieldeqs2}) is a set of the form $(\bar{g},\bar{\phi},K,\pi
)$, which under a Weyl transformation on the initial slice $M$, transforms
according to
\begin{align}
\label{constraintstrans2}%
\begin{split}
\bar{g}^{\prime}=e^{-h}\bar{g}, \;\;\;\;\;  &  ; \;\;\;\;\; \bar{\phi}^{\prime
}=\bar{\phi}-h,\\
K^{\prime}=e^{-\frac{h}{2}}K, \;\;\;\;\;  &  ; \;\;\;\;\; \pi^{\prime}=e^{\frac{h}{2}}\pi.
\end{split}
\end{align}
Clearly, given a solution of (\ref{wistconstraints}) and applying these
transformations to it, we get that, under the choice $h\doteq f(\cdot,0)$, the
transformed quantities $(\bar{g}^{*},\bar{\phi}^{*},K^{*},\pi^{*})$ satisfy
the system (\ref{transformedconstraints}), which is the constraint system for
the space-time equations (\ref{transformedsystem}).

\subsection*{The evolution problem}

We will now analyse the set of equations (\ref{fieldeqs2}). We will follow the
same line of argument typically applied to the vacuum system in GR
(Thorough reviews on this topic can be found in \cite{C-B1} and \cite{Ringstrom}).
With this in mind, we should begin by making the following remarks. First,
note that the geometric identities implied by the symmetries of the curvature
tensor are vital in this problem. We will need to use analogue identities in
our new setting. There is an easy way to show that there is a version of the
contracted Bianchi identities valid in this new context. To see this, just
note that given a Weyl integrable structure $(V,[g,\phi])$, for the
\textit{Riemannian element} in the class $(V,e^{-\phi}g,0)$ we have the usual
contracted Bianchi identities:
\[
e^{\phi}g^{\mu\nu}\nabla_{\mu}G_{\nu\sigma}=0,
\]
where $G_{\nu\sigma}$ is the Riemannian Einstein tensor associated with
$e^{-\phi}g$. Since the Einstein tensor is invariant under Weyl
transformations, then $G_{\nu\sigma}$ represents the Einstein tensor of all
the elements in the class $(V,[g,\phi])$. Thus, we easily see that
\begin{align*}
g^{\mu\nu}\nabla_{\mu}G_{\nu\sigma}(g,\phi)  &  =e^{-\phi}\Big(e^{\phi}%
g^{\mu\nu}\nabla_{\mu}G_{\mu\sigma}(e^{-\phi}g,0)\Big)\\
&  =0.
\end{align*}
Thus,
\[
g^{\mu\nu}\nabla_{\mu}G_{\nu\sigma}(g,\phi)=0,
\]
and this geometric identity is valid for every element in the class
$(V,[g,\phi])$. This identity clearly imposes a condition on the right-hand
side of (\ref{fieldeqs2}), which is typically referred to as a
\textit{conservation equation}. In order for the system (\ref{fieldeqs2}) to
be sensible, this geometric condition must be compatible with the equation
posed for $\phi$. It is a simple matter of computation to show that both
conditions are equivalent.

Secondly, we will rewrite (\ref{fieldeqs2}) merely in terms of the Ricci
tensor. In order to achieve this, notice that taking the trace of this
equation, we get
\begin{align*}
R=-\omega g^{\mu\nu}\nabla_{\mu}\phi\nabla_{\nu}\phi+\frac{n+1}{n-1}V(\phi).
\end{align*}
Using this information, we can rewrite our system in the following way:
\begin{align}
\label{fieldeqs3}%
\begin{split}
&  R_{\alpha\beta}=-\omega\nabla_{\alpha}\phi\nabla_{\beta}\phi+\frac{1}%
{n-1}g_{\alpha\beta}V(\phi),\\
&  g^{\alpha\beta}\nabla_{\alpha}\nabla_{\beta}\phi+\frac{1}{2\omega
}(V^{\prime}(\phi)+V(\phi))=0.
\end{split}
\end{align}

Now, some straightforward computations show that we can write down the Ricci
tensor as:
\[
R_{\alpha\beta}={}^{\circ}\!R_{\alpha\beta}+\frac{n-1}{2}{}^{\circ}%
\!\nabla_{\alpha}{}^{\circ}\!\nabla_{\beta}\phi+\frac{n-1}{4}{}^{\circ
}\!\nabla_{\alpha}\phi{}^{\circ}\!\nabla_{\beta}\phi+\frac{1}{2}g_{\alpha
\beta}(g^{\mu\nu}{}^{\circ}\!\nabla_{\mu}{}^{\circ}\!\nabla_{\nu}\phi
-\frac{n-1}{2}{}^{\circ}\!\nabla_{\mu}\phi{}^{\circ}\!\nabla^{\mu}\phi).
\]
It is worth noting that, again, using the invariance of the Ricci tensor under Weyl transformations, we can pick the Riemannian element in a class $(V,[g,\phi])$, for which its ``Riemannian" Ricci tensor coincide with the Ricci tensor of the class, \textit{i.e}, $R_{\alpha\beta}(g,\phi)={}^{\circ}\!R(e^{-\phi}g)$, and then make a conformal transformation, mapping $e^{-\phi}g\mapsto g$. Since it is well-known how ${}^{\circ}\!R(e^{-\phi}g)$ transforms under conformal transformations on the metric, that is, the relation between ${}^{\circ}\!R(e^{-\phi}g)$ and ${}^{\circ}\!R(g)$ is well-known, we can use such relation in $R_{\alpha\beta}(g,\phi)={}^{\circ}\!R(e^{-\phi}g)$, and we will get the right decomposition for the Ricci tensor of the Weyl connection of the class in terms of the Ricci tensor associated with the Riemannian connection of a particular element in the class, and the extra terms involving the scalar field.

Furthermore, a straightforward computation gives that
\begin{align*}
g^{\mu\nu}\nabla_{\mu}\nabla_{\nu}\phi=g^{\mu\nu}{}^{\circ}\!\nabla_{\mu}%
{}^{\circ}\!\nabla_{\nu}\phi-\frac{n-1}{2}{}^{\circ}\!\nabla_{\mu}\phi
{}^{\circ}\!\nabla^{\mu}\phi.
\end{align*}
Thus, the system (\ref{fieldeqs3}) is equivalent to the following set of
equations:
\begin{align}%
\begin{split}
\!{}^{\circ}\!R_{\alpha\beta}+\frac{n-1}{2}{}^{\circ}\!\nabla_{\alpha}%
{}^{\circ}\!\nabla_{\beta}\phi+\frac{n-1}{4}{}^{\circ}\!\nabla_{\alpha}\phi
{}^{\circ}\!\nabla_{\beta}\phi-\frac{1}{4\omega}g_{\alpha\beta}\big( V^{\prime
}(\phi)+V(\phi) \big)=  &  -\omega{}^{\circ}\!\nabla_{\alpha}\phi{}^{\circ
}\!\nabla_{\beta}\phi\\
&  +\frac{1}{n-1}g_{\alpha\beta}V(\phi),\label{fieldeqs4.1}%
\end{split}
\end{align}
\begin{align}
\!\!\!\!\!\!g^{\mu\nu}{}^{\circ}\!\nabla_{\mu}{}^{\circ}\!\nabla_{\nu}%
\phi-\frac{n-1}{2}{}^{\circ}\!\nabla_{\mu}\phi{}^{\circ}\!\nabla^{\mu}%
\phi+\frac{1}{2\omega}(V^{\prime}(\phi)+V(\phi))=0 . \label{fieldeqs4.2}%
\end{align}

The above field equations are posed for $(g,\phi)$. Looking at
(\ref{fieldeqs4.2}), it is clear that this is a
quasilinear wave equation which does not pose any
difficulties from the point of view of the theory of partial differential
equations. In contrast, in the set of equations (\ref{fieldeqs4.1}), the
second order term ${}^{\circ}\!\nabla_{\alpha}{}^{\circ}\!\nabla_{\beta}\phi$
brings some difficulties. If this term was not present, the same approach
towards the vacuum Einstein equations, which shows that such equations can be
written as a system of non-linear wave equations for the metric, could be
applied here, and the lower order terms in $\phi$ would not cause any
difficulties. This approach does not work here. It is interesting to remark
that the very same problem appears when studying, in the Jordan frame, the
Cauchy problem for Brans-Dicke theory and other scalar tensor theories. This
problem has been studied before, for instance, in \cite{Franceces}%
,\cite{Cocke},\cite{Noakes},\cite{Salgado1}. In \cite{Franceces} and
\cite{Cocke}, their considerations are limited to analytic initial data. In
\cite{Noakes}, the \textit{local} well-posedness of the initial value
formulation for the smooth case (or even $C^{k}$) is addressed. By local we
mean that the problem is studied within a coordinate system, which is chosen
by using wave-coodinates, and this condition is not invariant under coordinate
changes. Thus, the well-posedness is both local in space and in time. In
\cite{Salgado1} more general type of scalar-tensor theories are studied and
they are shown to have a well-posed initial value formulation by rewriting
them as a first order hyperbolic system where classical theorems on partial
differential equations (PDE) can be applied. In this approach, they also treat the
local in space and time problem, using a generalized wave-gauge condition.

In contrast to the above mentioned approaches, we will address the global in
space initial value formulation for the set of equations (\ref{fieldeqs4.1}%
)-(\ref{fieldeqs4.2}), and, furthermore, we will do this in such a way that
the invariance of the problem under Weyl transformations becomes transparent.
This will be done by following the standard approach to the global in space
problem in GR, but adapting our gauge conditions to our new setting. In this
direction, let us begin by defining the following properly Riemannian metric
on space-time:
\[
\hat{e}\doteq dt\otimes dt+e,
\]
where $e$ is some smooth Riemannian metric on $M$. Now define the following
vector field on space-time
\[
F^{\lambda}\doteq g^{\alpha\beta}(\Gamma_{\alpha\beta}^{\lambda}-\hat{\Gamma
}_{\alpha\beta}^{\lambda}),
\]
where $\Gamma_{\alpha\beta}^{\lambda}$ are connection components of the Weyl
connection associated with the space-time Weyl structure $(g,\phi)$, and
$\hat{\Gamma}_{\alpha\beta}^{\lambda}$ are the components of the Riemannian
connection associated with the metric $\hat{e}$. The fact that $F$ defines a
vector field on space-time, comes from the fact that the difference of two
connections defines a tensor field. We can rewrite this vector field in the
following way:
\[
F^{\lambda}={}^{\circ}\!F^{\lambda}+g^{\alpha\beta}(\Gamma_{\alpha\beta
}^{\lambda}-{}^{\circ}\Gamma_{\alpha\beta}^{\lambda}),
\]
where ${}^{\circ}\!F^{\lambda}\doteq g^{\alpha\beta}({}^{\circ}\!\Gamma
_{\alpha\beta}^{\lambda}-\hat{\Gamma}_{\alpha\beta}^{\lambda})$. The vector
field ${}^{\circ}\!F$ has been used to analyse the global in space problem in
GR \cite{C-B1}. Now, since $g^{\alpha\beta}(\Gamma_{\alpha\beta}^{\lambda}%
-{}^{\circ}\Gamma_{\alpha\beta}^{\lambda})=\frac{n-1}{2}{}^{\circ}%
\nabla^{\lambda}\phi$, we get that
\begin{align}\label{Weylgauge2}
F^{\lambda}={}^{\circ}\!F^{\lambda}+\frac{n-1}{2}{}^{\circ}\nabla^{\lambda
}\phi.
\end{align}
At this point, we will take advantage of some useful expressions known from
the Cauchy problem in GR. For instance, it is known that the Ricci tensor
${}^{\circ}\!R_{\alpha\beta}$ has the following decomposition
(see, for instance, chapter 6 in \cite{C-B1}):
\[%
\begin{split}
{}^{\circ}\!R_{\alpha\beta} &  =R_{\alpha\beta}^{(\hat{e})}+\frac{1}%
{2}(g_{\alpha\lambda}D_{\beta}{}^{\circ}\!F^{\lambda}+g_{\beta\lambda
}D_{\alpha}{}^{\circ}\!F^{\lambda}),\\
R_{\alpha\beta}^{(\hat{e})} &  =-\frac{1}{2}g^{\lambda\mu}D_{\lambda}D_{\mu
}g_{\alpha\beta}+f_{\alpha\beta}(g,Dg),
\end{split}
\]
where $D$ represents the $\hat{e}$-covariant derivative and $f_{\alpha\beta}$
has smooth dependence on both $g$ and $Dg$. Using this expressions together
with (\ref{Weylgauge2}), we can rewrite
\begin{align}\label{gaugericci}
{}^{\circ}\!R_{\alpha\beta}=R_{\alpha\beta}^{(\hat{e})}+\frac{1}{2}%
(g_{\alpha\lambda}D_{\beta}F^{\lambda}+g_{\beta\lambda}D_{\alpha}F^{\lambda
})-\frac{n-1}{4}(g_{\alpha\lambda}D_{\beta}g^{\lambda\mu}D_{\mu}\phi
+g_{\beta\lambda}D_{\alpha}g^{\lambda\mu}D_{\mu}\phi)-\frac{n-1}{2}D_{\alpha
}D_{\beta}\phi.
\end{align}
Now, the idea is to rewrite (\ref{fieldeqs4.1}) in terms of the $\hat{e}%
$-covariant derivatives, and to use the previous decomposition for the
Riemannian Ricci tensor. The interesting thing is that the last term in
(\ref{gaugericci}) will exactly cancel the problematic term in
(\ref{fieldeqs4.1}). In this way, we get the following system:
\begin{align}\label{fullsys}
\begin{split}
&  {}^{r}\!R_{\alpha\beta}-\rho_{\alpha\beta}+\frac{1}{2}(g_{\alpha\lambda
}D_{\beta}F^{\lambda}+g_{\beta\lambda}D_{\alpha}F^{\lambda})=0,\\
&  g^{\alpha\beta}D_{\alpha}D_{\beta}\phi-{}^{\circ}\!F^{\sigma}D_{\sigma}%
\phi-\frac{n-1}{2}D_{\sigma}\phi D^{\sigma}\phi+\frac{1}{2\omega
}\big(V^{\prime}(\phi)+V(\phi)\big)=0,
\end{split}
\end{align}
where
\[%
\begin{split}
{}^{r}\!R_{\alpha\beta} &  \doteq R_{\alpha\beta}^{(\hat{e})}-\frac{n-1}%
{4}(g_{\alpha\lambda}D_{\beta}g^{\lambda\mu}D_{\mu}\phi+g_{\beta\lambda
}D_{\alpha}g^{\lambda\mu}D_{\mu}\phi)-\frac{n-1}{2}\Delta{}^{\circ}%
\Gamma_{\alpha\beta}^{\sigma}D_{\sigma}\phi\\
&  +\frac{n-1}{4}D_{\alpha}\phi D_{\beta}\phi-\frac{1}{4\omega}g_{\alpha\beta
}\big(V^{\prime}(\phi)+V(\phi)\big),\\
\rho_{\alpha\beta} &  \doteq-\omega D_{\alpha}\phi D_{\beta}\phi+\frac{1}%
{n-1}g_{\alpha\beta}V(\phi),\\
\Delta{}^{\circ}\Gamma_{\alpha\beta}^{\sigma} &  \doteq{}^{\circ}%
\!\Gamma_{\alpha\beta}^{\sigma}-\hat{\Gamma}_{\alpha\beta}^{\sigma}.
\end{split}
\]
Following an approach analogous to the standard approach to the Einstein
equations, we will consider the reduced system:
\begin{align}
&  {}^{r}\!R_{\alpha\beta}-\rho_{\alpha\beta}=0,\label{reducedsys1}\\
&  g^{\alpha\beta}D_{\alpha}D_{\beta}\phi-{}^{\circ}\!F^{\sigma}D_{\sigma}%
\phi-\frac{n-1}{2}D_{\sigma}\phi D^{\sigma}\phi+\frac{1}{2\omega
}\big(V^{\prime}(\phi)+V(\phi)\big)=0.\label{reducedsys2}%
\end{align}
Using the expression for $R_{\alpha\beta}^{(\hat{e})}$ in (\ref{reducedsys1}),
we see that (\ref{reducedsys1}) and (\ref{reducedsys2}) form a system of
hyperquasilinear wave equations for $(g,\phi)$, thus, given appropriate
initial data, the system has one and only one solution on $M\times\lbrack
0,T)$, for some $T>0$. By appropriate initial data, we mean initial data belonging to some appropriate functional spaces. The regularity of the solution will be directly tied to this choice. In particular, for smooth data, we get smooth solutions. For a detailed review of these properties, and low regularity results, see, for instance, \cite{C-B1}.

Now, since $T_{\alpha\beta}=\rho_{\alpha\beta}-\frac{1}{2}g_{\alpha\beta
}g^{\mu\nu}\rho_{\mu\nu}$, we get that, given a solution $(g,\phi)$ for
(\ref{reducedsys1})-(\ref{reducedsys2}), its Einstein tensor satisfies the
following identity:
\begin{align*}
G_{\alpha\beta}-T_{\alpha\beta}=\frac{1}{2}(g_{\alpha\lambda}D_{\beta
}F^{\lambda}+g_{\beta\lambda}D_{\alpha}F^{\lambda}-g_{\alpha\beta}D_{\lambda
}F^{\lambda}).
\end{align*}
Now, the contracted Bianchi identities for Weyl's connection together with
(\ref{reducedsys2}), imply that the divergence of the previous expression
vanishes. This gives us a geometric identity satisfied by $F$. Some
straightforward computations give us that $F$ satisfies the following
homogeneous system of linear wave equations:
\begin{align}
g^{\alpha\mu}D_{\alpha}D_{\mu}F^{\lambda}+A^{\alpha\lambda}_{\nu}%
(g,\phi,Dg,D\phi)D_{\alpha}F^{\nu}+g^{\beta\lambda}R_{\mu\beta}(\hat{e}%
)F^{\mu}=0,
\end{align}
where $A^{\alpha\lambda}_{\nu}$ is a $(2,1)$-tensor field depending on
$(g,\phi,Dg,D\phi)$. Uniqueness of solutions for such a system gives us that,
if $F(\cdot,0)=0$ and $\partial_{t}F(\cdot,0)=0$, then $F\equiv0$. Thus, if we
can guarantee that the initial data for (\ref{reducedsys1})-(\ref{reducedsys2}%
) satisfy these conditions, then the solution to this system actually
satisfies the full system of equations (\ref{fullsys}).

\begin{lemma}
Suppose that $(g,\phi)$ satisfy the reduced system of equation
(\ref{reducedsys1})-(\ref{reducedsys2}) and that the initial data satisfies
$F(\cdot,0)=0$. Then $\partial_{t}F(\cdot,0)=0$ iff the constraint equations
(\ref{wistconstraints}) are satisfied.
\end{lemma}

\begin{proof}
We have the following decomposition:
\begin{align*}
G_{\alpha\beta}={}^{r}\!G_{\alpha\beta}+\frac{1}{2}(g_{\alpha\lambda}D_\beta F^{\lambda}+g_{\beta\lambda}D_{\alpha}F^{\lambda}-g_{\alpha\beta}D_{\lambda}F^{\lambda}),
\end{align*}
where ${}^{r}\!G_{\alpha\beta}={}^{r}\!R_{\alpha\beta}-\frac{1}{2}g_{\alpha\beta}(g^{\mu\nu}\:{}^{r}\!R_{\alpha\beta})$.
Note that if $F|_{t=0}=0$, then $D_iF|_{t=0}=0$. Then we have that:
\begin{align*}
G_{\alpha 0}|_{t=0}={}^{r}G_{\alpha 0}|_{t=0}+\frac{1}{2}(g_{\alpha i}D_0F^i-N^2D_\alpha F^0)|_{t=0}.
\end{align*}
Thus,
\begin{align*}
(G_{\alpha 0}-T_{\alpha 0})|_{t=0}=({}^{r}G_{\alpha 0}-T_{\alpha 0})|_{t=0}+\frac{1}{2}(g_{\alpha i}D_0F^i-N^2D_\alpha F^0)|_{t=0}.
\end{align*}
Then, since $(g,\phi)$ solve the reduced system, we get that
\begin{align*}
(G_{\alpha 0}-T_{\alpha 0})|_{t=0}=\frac{1}{2}(g_{\alpha i}D_0F^i-N^2D_\alpha F^0)|_{t=0}.
\end{align*}
Thus we see that the constraint equations are satisfied iff
\begin{align*}
0=(g_{\alpha i}D_0F^i-N^2D_\alpha F^0)|_{t=0}.
\end{align*}
Settig $\alpha=0$, the previous expression gives $\partial_tF^0=0$ and setting $\alpha=j$ gives $\partial_tF^j=0$.
\end{proof}

Using this lemma, we see that if we give an initial data set satisfying the
constraint equations plus the \textit{gauge condition} $F|_{t=0}=0$, then the
system (\ref{reducedsys1})-(\ref{reducedsys2}) will have a unique solution,
which, in fact, solves the full system (\ref{fieldeqs4.1})-(\ref{fieldeqs4.2}%
). Let us see that we can always pick initial data satisfying the gauge condition.

Our geometric initial data set for (\ref{reducedsys1})-(\ref{reducedsys2}) is
$(M,\bar{g},\bar{\phi},K,\pi)$. We need to provide initial data for $(g,\phi
)$. It is clear that $\bar{g},\bar{\phi}$ determine the initial data
$g_{ij}|_{t=0},\phi|_{t=0}$. Also that, once we fix the initial data for the
lapse and shift, $K$ and $\pi$ determine the initial data $\partial_{t}%
g_{ij}|_{t=0}$ and $\partial_{t}\phi|_{t=0}$. Then, we freely set
$g_{00}|_{t=0}=-1$ and $g_{0i}|_{t=0}=0$, which amounts to setting
$N|_{t=0}=1$ and $\beta|_{t=0}=0$. We still need initial the data
$\partial_{t}g_{00}|_{t=0},\partial_{t}g_{0i}|_{t=0}$. These initial data will
come from the solving the gauge condition. Note that
\begin{align*}
F^{\lambda}|_{t=0}=\{g^{\alpha\beta}{}^{\circ}\!\Gamma^{\lambda}_{\alpha\beta
}+\frac{n-1}{2}{}^{\circ}\nabla^{\lambda}\phi-g^{\alpha\beta}\hat{\Gamma
}^{\lambda}_{\alpha\beta}\}|_{t=0}.
\end{align*}
Then, since we have that
\begin{align*}%
\begin{split}
g^{\alpha\beta}{}^{\circ}\!\Gamma^{0}_{\alpha\beta}|_{t=0}  &  =\frac{1}%
{2}\partial_{t}g_{00}|_{t=0}-\bar{g}^{ij}{}^{\circ}\!K_{ij},\\
g^{\alpha\beta}{}^{\circ}\!\Gamma^{i}_{\alpha\beta}|_{t=0}  &  =-\overline
{g}^{ij}\partial_{t}g_{j0}|_{t=0}+\overline{g}^{kl}{}^{\circ}\!\bar{\Gamma
}^{i}_{kl},
\end{split}
\end{align*}
we can solve the algebraic equations $F^{\lambda}|_{t=0}=0$ for $\partial
_{t}g_{00}|_{t=0},\partial_{t}g_{0i}|_{t=0}$. With these choices for
$\partial_{t}g_{00}|_{t=0},\partial_{t}g_{0i}|_{t=0}$ we complete the initial
data for the reduced system. Then, the theory of hyperbolic partial
differential equations guarantees the existence of a unique solution for this
system, which will be as regular as the initial data. Thus, we have the
following theorem:

\begin{thm}
[\textbf{Existence in a frame}]\label{existence1} Given initial data
$(M,\bar{g},\bar{\phi},K,\pi)$ satisfying the constraint equations
(\ref{wistconstraints}), there is a Cauchy development of this initial data
set into a space-time $(V,g,\phi)$ satisfying the field equations
(\ref{fieldeqs4.1})-(\ref{fieldeqs4.2}).
\end{thm}

At this point, it is worth noting that the gauge condition we have chosen, \textit{i.e}, $F|_{t=0}=0$, is similar to the one proposed in \cite{Salgado1}. The difference between the two gauges is that, while in \cite{Salgado1} the gauge condition is cleverly chosen so as to apply to quite general scalar-tensor theories, this condition, despite its simplicity, reveals a non-trivial insight presented by the author, and it is presented as a local condition,  meaning that is represents a choice for a preferred coordinate system. On the other hand, our gauge condition comes as a natural consequence of following the standard approach to the global in space formulation of the Cauchy problem and exploiting the specific geometric properties of WGSTT.

\bigskip

We will now show that the problem is \textit{invariant under Weyl
transformations}. By this we mean that if we make a space-time Weyl
transformation, then, the space-time described in this other frame results as
the evolution of initial data \textit{equivalent} to $(\overline{g}%
,\overline{\phi},K,\pi)$. And conversely, given an \textit{equivalent} initial
data set to $(\overline{g},\overline{\phi},K,\pi)$, then this initial data set
admits a development into a space-time \textit{equivalent} to the one
generated by $(\overline{g},\overline{\phi},K,\pi)$.

Consider that we have an initial data set $(\overline{g},\overline{\phi}%
,K,\pi)$ satisfying the constraint equations (\ref{wistconstraints}) for the
system (\ref{fieldeqs2}). Then, using \textbf{Theorem \ref{existence1}}, we
know that this initial data set admits a development into a space-time
$(V,g,\phi)$ satisfying the field equations (\ref{fieldeqs2}). Then, perform a
space-time Weyl transformation of the form
\begin{align}\label{consistency1}
g^{\prime}=e^{-f}g, \;\;\;\; \phi^{\prime}=\phi-f
\end{align}
We know that $(g^{\prime},\phi^{\prime})$ will satisfy the system of equations
(\ref{transformedsystem}). Furthermore, since the gauge condition $F=0$ is
Weyl-invariant, \textit{i.e.}, if it satisfied in one frame, then it is
satisfied in all the equivalence class, then applying the same procedure
described above for the system (\ref{fieldeqs2}) to the transformed system
(\ref{transformedsystem}), we see that (\ref{transformedsystem}) is actually
equivalent to the set of reduced equations that result from setting
$F(g^{\prime},\phi^{\prime})=0$. Thus, we see that $(g^{\prime},\phi^{\prime
})$ actually solve a system of non-linear wave equations in the space-time
metric $g^{\prime}$ analogous to (\ref{reducedsys1})-(\ref{reducedsys2}).
Furthermore, this system takes initial data satisfying the constraints
(\ref{transformedconstraints}). This shows that, considering a Weyl transformation on the initial data $(\overline{g},\overline{\phi},K,\pi)$ of the form
\begin{align}\label{consistency2}
\bar{g}^{\prime}=e^{-h}\bar{g}, \;\;\;\; \bar{\phi}^{\prime
}=\bar{\phi}-h, \;\;\;\; K^{\prime}=e^{-\frac{h}{2}}K, \;\;\;\; \pi^{\prime}=e^{\frac
{h}{2}}\pi,
\end{align}
and setting $h\doteq f(\cdot,0)$, then $(V,g^{\prime},\phi^{\prime})$ results
as the evolution of the initial data set $(\overline{g}^{\prime}%
,\overline{\phi}^{\prime},K^{\prime},\pi^{\prime})$, which is equivalent to
$(\overline{g},\overline{\phi},K,\pi)$.

Conversely, if we make a Weyl transformation on the initial data
$(\overline{g},\overline{\phi},K,\pi)$ of the form (\ref{consistency2}), then
the transformed initial data will satisfy the system
(\ref{transformedconstraints}) if we consider $f|_{t=0}=h$. Thus, if we pick a
function $f$ on space-time satisfying $f|_{t=0}=h$, and then we perform a
space-time Weyl transformation on $(g,\phi)$ of the form (\ref{consistency1}),
we see that $(g^{\prime},\phi^{\prime})$ on $M_{0}$ satisfy the initial
conditions $(\bar{g}^{\prime},\bar{\phi}^{\prime},K^{\prime},\pi^{\prime})$.
Thus, any initial data set equivalent to $(\overline{g},\overline{\phi}%
,K,\pi)$ has an evolution into a space-time equivalent to $(V,g,\phi)$. We
summarize this with the following theorem.

\begin{thm}
Given an initial data set $(M,\bar{g},\bar{\phi},K,\pi)$ satisfying the
constraint equations (\ref{wistconstraints}), then the Cauchy problem for the
WGSTT is well-posed.
\end{thm}

\qed

\subsection*{Geometric uniqueness}

We will finally discuss a subtle point, which in the context of GR is referred
to as \textit{geometric uniqueness}. This is related to the arbitrariness in
the choice of the initial data for the lapse and shift. As we have sated
above, this is a gauge freedom which is already present in GR. In our present
context, besides this gauge freedom, we have added Weyl transformations. We
have already dealt with the latter, and we will now deal with the former.

The question that naturally arises in this context, is that different choices
for the initial values of the lapse and shift will provide different initial
data sets for our PDE system. Thus, we will get \textit{different} solutions
for each case. If this were the case, then the lapse and shift would not
actually be working as gauge variables. It is a very well-known fact that the
solution to this problem comes by taking into account that in space-time we
regard \textit{diffeomorphic} structures as equivalent. That is, given the
group of diffeomorphisms on a given space-time, such group defines an action
on the different tensor fields defined on space-time, acting via pull-back,
and such an action defines an equivalence relation. With this in mind, within
the context of GR, it can be shown, for instance for vacuum initial data sets,
that two different choices of initial data for the lapse and shift give rise
to initial data sets whose associated space-time evolutions are diffeomorphic,
and thus \textit{equivalent}
(see, for instance, chapter 6 in \cite{C-B1}). We will now
show that the very same procedure works in this context. In order to present
this result we will first need to introduce some terminology.

In order to show that different choices of initial data for the lapse and
shift give rise to diffeomorphic space-times, we will actually show how to
construct such diffeomorphisms. For this purpose we need to introduce the
notion of a \textit{wave map}.

Given two pseudo-Riemannian manifolds $(V,g)$ and $(M,h)$, let $u:V\mapsto M$
be a smooth map. We can consistently define the vector bundle
$E\xrightarrow{\pi} V$, over $V$, whose typical fibre over a point $x\in V$ is
given by $T^{*}_{x}V\otimes T_{u(x)}M$. That is,
\begin{align*}
E=\coprod_{x\in V} T^{*}_{x}V\otimes T_{u(x)}M.
\end{align*}

We can now endow $E$ with a natural connection, denoted by $\nabla^{E}$, by
considering local charts $(\mathcal{W},X)$, $\mathcal{W}\subset V$, and
$(\mathcal{\tilde{W}},Y)$, $\mathcal{\tilde{W}}\subset M$, and applying
Leibniz rule to a local section of $E$ of the form $f=f^{A}_{\alpha}%
dx^{\alpha}\otimes\partial_{y^{A}}$. Now we define the action of the
connection on an element of $T^{*}_{x}V$ using the coefficients of the
Riemannian connection at $x$ of the metric $g$, while the coefficients acting
on $T_{u(x)}M$ are the pull-back by $u$ of the Riemannian connection 1-form of
the metric $h$ at $u(x)$. Some simple computations then give us the
following:
\begin{align}
\nabla^{E}_{\beta}f^{A}_{\alpha}=\partial_{\beta}f^{A}_{\alpha}-{}^{g}%
\!\Gamma^{\sigma}_{\alpha\beta}f^{A}_{\sigma}+\partial_{\beta}u^{C}{}%
^{h}\!\Gamma^{A}_{BC}f^{B}_{\alpha},
\end{align}
where the connection coefficients ${}^{g}\!\Gamma^{\sigma}_{\alpha\beta}$ and
${}^{h}\!\Gamma^{A}_{BC}$ are the Christoffel symbols associated with the
Riemannian connections of $g$ and $h$ respectively.


Furthermore, notice that the \textquotedblleft\textit{gradient}" $\partial u$
of the map $u$, given in local coordinates by $\partial_{\alpha}u^{A}$,
defines a section of the vector bundle $E$. With all this notations, it is
said that the mapping $u:(V,g)\mapsto(M,h)$ is a \textbf{wave map} if it
satisfies the following PDE:
\[
\mathrm{tr_{g}}\nabla^{E}\partial u=0.
\]
For a detailed review on wave maps, see \cite{C-B2}.

Returning to the problem of geometric uniqueness for the initial value
formulation for the WGSTT, given two smooth solutions $(V_{1},g_{1},\phi_{1})$
and $(V_{2},g_{2},\phi_{2})$ to the Cauchy problem with the same geometric
initial data $(M,\bar{g},\bar{\phi},K,\pi)$, we will construct a
diffeomorphism from a neighbourhood $U_{1}$ of $M_{0}\subset V_{1}$ onto a
neighbourhood $U_{2}$ of $M_{0}\subset V_{2}$. For this purpose we need the
following lemma.

\begin{lemma}
Suppose that $(M\times\mathbb{R},g)$ is a globally hyperbolic manifold with
$g$ smooth. There exists a wave map $f$, which is a diffeomorphism from a
strip $W_{\tau}\doteq M\times(-\tau,\tau)$ onto a neighbourhood of
$M\times\{0\}$ in $M\times\mathbb{R}$, such that $f$ takes the following
initial values:
\begin{align}
\label{ID}f(\cdot,0)=Id, \;\;\; ; \;\;\; \partial_{t}f(\cdot,0)=(a,b),
\end{align}
where $a>0$ is a specified scalar on $M$ and $b$ a specified tangent vector to
$M$. \qed

\end{lemma}

For a short proof of this lemma see \cite{C-B1}, and for a detailed treatment
of this problem see \cite{C-B2}. We can prove the following:

\begin{lemma}
Let $f$ be the unique wave map corresponding to the data (\ref{ID}). Such map
defines a diffeomorphism between neighbourhoods of $M_{0}$ in $V$. Let
$\underline{g}\doteq f^{*}g$ and $\underline{\phi}\doteq f^{*}\phi$ be the
metric and scalar field in the wave map pull back by $f$ of a Lorentzian
metric $g$ and a scalar field $\phi$ which result as the evolution of a
specified initial data set $(\bar{g},\bar{\phi},K,\pi)$. Then, it is possible
to choose ``$a$" and ``$b$" such that the initial values of the metric
$\underline{g}$ and the scalar field $\underline{\phi}$, and their first
time-derivatives, take preassigned specified values, depending only on the
geometric data $(\bar{g},\bar{\phi},K,\pi)$.
\end{lemma}

\begin{proof}
Notice that some of the initial data for $(\underline{g},\underline{\phi})$ are already specified by the geometric initial data set $(\bar{g},\bar{\phi},K,\pi)$. For instance, recalling that two Weyl manifolds $(M_1,g_1,\phi_1)$ and $(M_2,g_2,\phi_2)$ are called isometric if there is a diffeomorphism $f:M_1\mapsto M_2$ such that $g_1=f^{*}g_2$ and $\phi_1=f^{*}\phi_2$, then, by means of the wave map diffeomorphism we are defining an isometry between the two Weyl manifolds $(V,g,\phi)$ and $(V,\underline{g},\underline{\phi})$, thus all the intrinsically geometric objects will be unaffected. That is, the Weyl connection of the two structures is the same, and so is the extrinsic curvature of $M_0$ (just think of the diffeomorphism as generating coordinate transformations). Thus, we already have the initial data $\underline{\phi}(\cdot,0)=\phi(\cdot,0)$ and $\underline{K}=K$. Now, the diffeomorphism gives us the freedom to specify the initial values of the lapse $\underline{N}$ and shift $\underline{\beta}$. In particular we are interested in setting $\underline{N}(\cdot,0)=1$ and $\underline{\beta}(\cdot,0)=0$. Taking into account (\ref{ID}), we get the following relations:
\begin{align*}
\underline{g}(\partial_t,\partial_t)|_{t=0}&=g(f_{*}\partial_t,f_{*}\partial_t)|_{t=0},\\
\underline{g}(\partial_t,\partial_i)|_{t=0}&=g(f_{*}\partial_t,f_{*}\partial_i)|_{t=0}=g(f_{*}\partial_t,\partial_i)|_{t=0},\\
\underline{g}(\partial_i,\partial_j)|_{t=0}&=g(f_{*}\partial_i,f_{*}\partial_j)|_{t=0}=\bar{g}(\partial_i,\partial_j).
\end{align*}
In order to satisfy our prescribed conditions on the initial data for the shift and lapse, we get that
\begin{align*}
f_{*}\partial_t(\cdot,0)=\frac{1}{N}e_0|_{t=0}=\frac{1}{N}(\partial_t-\beta)|_{t=0},
\end{align*}
from which we deduce that $\partial_tf(\cdot,0)=(\frac{1}{N}|_{t=0},-\frac{1}{N}\beta|_{t=0})$, which implies the choices $a=\frac{1}{N}|_{t=0}$ and $b=-\frac{1}{N}\beta|_{t=0}$. Now, by means of the relation between $\partial_{t}\underline{g}_{ij}$ and $\underline{K}_{ij}$, we can pick the initial data $\partial_{t}\underline{g}_{ij}|_{t=0}$. Similarly, once we have set $(a,b)$ in terms of the initial data for $\underline{N}$ and $\underline{\beta}$, we get the initial data for $\partial_t\underline{\phi}|_{t=0}$. Finally, the initial data for $\partial_t\underline{g}_{0\alpha}|_{t=0}$ is again chosen by solving the initial wave gauge condition for the vector field $F$.
\end{proof}

\begin{thm}
[\textbf{Short-time geometric uniqueness}]\label{geouniq} Let $(V_{1}%
,[g_{1},\phi_{1}])$ and $(V_{2},[g_{2},\phi_{2}])$ be two smooth solutions of
the Cauchy problem for the vacuum WGSTT satisfying the initial data
$(M,[\bar{g},\bar{\phi},\bar{K},\bar{\pi}])$. There exists an isometry from
$(U_{1},[g_{1},\phi_{1}])$ onto $(U_{2},[g_{2},\phi_{2}])$, where $U_{1}$ and
$U_{2}$ are neighbourhoods of $M_{0}$, respectively in $V_{1}$ and $V_{2}$.
\end{thm}

\begin{proof}
Consider the representatives of the solution $(V_1,[g_1,\phi_1])$ and $(V_1,[g_2,\phi_2])$, given by $(V_1,g_1,\phi_1)$ and $(V_1,g_2,\phi_2)$, which satisfy the system (\ref{fieldeqs2}) and take the initial data $(M,\bar{g},\bar{\phi},\bar{K},\bar{\pi})$. Using the above lemmas, we know that there are diffeomorphisms $f_1$ and $f_2$ from neighbourhoods $U_1$ and $U_2$ of $M$ in $V_1$ and $V_2$, respectively, onto some neighbourhoods of $M$, such that $({f^{-1}_1}^{*}g_1,{f^{-1}_1}^{*}\phi_1)$ and $({f^{-1}_2}^{*}g_2,{f^{-1}_2}^{*}\phi_2)$ take the same initial data on $M_0$. Also, both $({f^{-1}_i}^{*}g_i,{f^{-1}_i}^{*}\phi_i)$, $i=1,2$, satisfy the reduced system (\ref{reducedsys1})-(\ref{reducedsys2}) and the gauge condition, since these are tensor relations. Thus, by uniqueness of solutions for the reduced system, both solutions agree. Thus, the composition $f_2^{-1}\circ f_1$ defines an isometry from $(U_1,g_1,\phi_1)$ onto $(U_2,g_2,\phi_2)$, where $U_1$ and $U_2$ are neighbourhoods of $M_0$ in $V_1$ and $V_2$ respectively. Thus, the equivalence classes $(U_1,[g_1,\phi_1])$ and $(U_2,[g_2,\phi_2])$ are isometric.
\end{proof}

With this theorem we end the local in time study of the initial value
formulation of the WGSTT without non-geometric (matter) sources. The final
result is that the Cauchy problem for such models is well-posed within the
physical interpretation proposed for those theories. That is, not only the PDE
systems have a well-posed initial value formulation, but the uniqueness holds
within the physical equivalence class. In this sense, when we speak of the
physical equivalence class $(V,[g,\phi])$, we should consider as equivalent
not only elements linked by the action of Weyl transformations, but also by
the action of the diffeomorphism group of $V$, which act on tensor fields via pull-back.

\section{Discussion}

We would like to conclude making a few remarks about the results presented
above. First, we have been able to show that the WGSTT presented in the first
sections, have a well-posed initial value formulation for initial data
satisfying a system of constraint equations. Furthermore, it is an interesting
fact that the Weyl structure results as the evolution of initial data, where
the space-time structure, in each frame, is generated as a solution of a
system of non-linear wave equations in the space-time metric. This shows that
in each frame we have that the speed of propagation of the gravitational
interaction, now described by the whole Weyl structure, is determined by the
null cones (of any) of the space-time metrics in the conformal class, which is
an invariant property in Weyl structure.

We should stress that the results presented above should be interpreted
according to \textbf{Theorem \ref{geouniq}}, where we have shown the
short-time geometric uniqueness for the initial value problem, taking into
account both the invariance under the action of Weyl transformations and
isometries. Again, it is worth stressing the relation between this problem and
the initial value formulation for standard scalar-tensor theories. In that
context, well-posedness has been studied in each frame independently, without
analysing the (possible) equivalence between space-times evolving from
equivalent geometric initial data. Such result should be of interest for
someone willing to adopt Dicke's interpretation regarding the equivalence of
both frames \cite{Dicke:1961gz}. In this line, by considering a geometric
interpretation for such theories, where by means of Weyl integrable structures
we can make precise mathematical sense of the physical equivalence of
different frames, we have not only shown well-posedness in each frame
independently, but also shown that equivalent initial data sets evolve into
equivalent space-time structures. Furthermore, we have provided a global in
space result, which, as far as we aware, is an issue which had not been
addressed outside the Riemannian (Einstein) frame.

We leave as a future research perspective the study of the evolution problem
in the presence of matter fields. Furthermore, it would be essential to
analyse the well-posedness of the system of constraint equations, and all the
related problems to this issue.

\section*{Acknowledgements}

\noindent The authors would like to thank CNPq and CAPES for financial
support. R. A. and C. R. would also like to thank CLAF for partial financial support.

\bigskip
\noindent
We thank the referee for valuable comments and suggestions.


\bigskip



\newpage

\begin{thebibliography}{99}                                                                                               %




\bibitem {Romero1}T. S. Almeida, M. L. Pucheu, C. Romero, and J.B. Formiga,
\textit{From Brans-Dicke gravity to a geometrical scalar-tensor theory}, Phys.
Rev. D, 89, 064047 (2014).

\bibitem {Salim-Poulis}F. P. Poulis and J. M. Salim, \textit{Weyl geometry and
gauge-invariant gravitation}, Int. J. Mod. Phys. D, 23, 1450091 (2014).

\bibitem {Romero2}C. Romero, J B Fonseca-Neto and M L Pucheu, \textit{General
relativity and Weyl geometry }, Class. Quantum Grav., 29, no. 15, 155015 (2012).

\bibitem {CBPF}. M. Novello, L.A.R. Oliveira, J.M. Salim, E. Elbas, Int. J.
Mod. Phys. \textbf{D1} (1993) 641-677.\ J. M. Salim and S. L. Saut\'{u},
Class. Quant. Grav. \textbf{13}, 353 (1996). H. P. de Oliveira, J. M. Salim
and S. L. Saut\'{u}, Class.Quant.Grav. \textbf{14},\textbf{\ }2833 (1997). V.
Melnikov, \textit{Classical Solutions in Multidimensional Cosmology} in
Proceedings of the VIII Brazilian School of Cosmology and Gravitation II
(1995), edited by M. Novello (Editions Fronti\`{e}res) pp. 542-560, ISBN
2-86332-192-7. K.A. Bronnikov, M.Yu. Konstantinov, V.N. Melnikov, Grav.Cosmol.
\textbf{1}, 60\textbf{\ }(1995). J. Miritzis, Class. Quantum .Grav.
\textbf{21}, 3043 (2004). J. Miritzis, J.Phys. Conf. Ser . \textbf{8},131
(2005). J.E.M. Aguilar and C. Romero, Found. Phys.\textbf{ 39} (2009)1205;
J.E.M. Aguilar and C. Romero, Int. J. Mod. Phys. A \textbf{24,} 1505 (2009).
J. Miritzis, Int. J. Mod. Phys. D \textbf{22}, 1350019 (2013). R. Vazirian, M. R. Tanhayi and Z. A. Motahar, Adv. High Energy Physics \textbf{7}, 902396 (2015).

\bibitem{Scholz1} E. Scholz, \textit{MOND-Like Acceleration in IntegrableWeyl
Geometric Gravity}, Found. Phys., 46, 176–208 (2016).

\bibitem {EPS}J. Ehlers, F. Pirani, and A. Schild, Gen. Rel. Grav., 44, Issue
6, 1587 (2012)\textit{.}

\bibitem {ADR} R. Avalos, F. Dahia and C. Romero, \textit{A Note on the Problem of Proper Time in Weyl Space–Time}, Found. Phys.,  48, 253-270 (2018).











\bibitem {Romero3}M. L. Pucheu, C. Romero, M. Bellini and J. E. M. Aguilar,
\textit{Gauge invariant fluctuations of the metric during inflation from a new
scalar-tensor Weyl-integrable gravity model} , Phys.Rev. D, 94, 6, 064075 (2016).

\bibitem{Romero4} M.L. Pucheu, F.A.P. Alves Junior, A. B. Barreto, C. Romero, \textit{Cosmological models in Weyl geometrical scalar-tensor theory}, Phys.Rev. D, 94, 6, 064010 (2016).

\bibitem {Scholz2}E. Scholz, \textit{The unexpected resurgence of Weyl
geometry in late 20-th century physics}, arXiv:1703.03187 (2017).

\bibitem {Iarley}I. P. Lobo, \textit{On the physical interpretation of
non-metricity in Brans-Dicke gravity}, Int. J. Geom. Methods Mod. Phys., 18, 1850138 (2017).

\bibitem {Folland}G. Folland, \textit{Weyl manifolds}, J. Diff. Geom., 4,
145-153 (1970).

\bibitem {Faraoni:2006fx}V. Faraoni and S. Nadeau, \textit{The (pseudo)issue
of the conformal frame revisited}, Phys.\ Rev.\ D, 75, 023501 (2007).

\bibitem {Quiros:2011wb}I. Quiros, R. Garcia-Salcedo, J. E. Madriz Aguilar and
T. Matos, \textit{The conformal transformation's controversy: what are we
missing?}, Gen.\ Rel.\ Grav.\ \textbf{45}, 489 (2013).

\bibitem {Dicke:1961gz}R. H. Dicke, \textit{Mach's principle and invariance
under transformation of units}, Phys.\ Rev.\ \textbf{125}, 2163 (1962).



\bibitem {Oneill}B. O'Neill, Semi-Riemannian Geometry With Applications to
Relativity. Academic Press. (1983). Chapter 4.

\bibitem {C-B1}Yvonne Choquet-Bruhat, \textit{General Relativity and the
Einstein equations}, Oxford University Press Inc., New York (2009).

\bibitem {Ringstrom}H. Ringstr\"om, \textit{The Cauchy Problem in General
Relativity}, European Mathematical Society, Germany (2009).





\bibitem {Franceces}P. Teyssandier and Ph. Tourrenc, \textit{The Cauchy
problem for the R+R2 theories of gravity without torsion}, Journ of Math.
Phys., 24, 2793 (1983).

\bibitem {Cocke}W. J. Cocke and Jeffrey M. Cohen, \textit{Cauchy Problem in
the Scalar-Tensor Gravitational Theory}, Journ. of Math. Phys., 9, 971 (1968).

\bibitem {Noakes}D. R. Noakes, \textit{The initial value formulation of higher
derivative gravity}, Journ. of Math. Phys., 24, 1846 (1983).

\bibitem {Salgado1}M. Salgado, \textit{The Cauchy problem of scalar-tensor
theories of gravity}, Class. Quantum Grav., 23, 4719-4741 (2006).

\bibitem {C-B2}Y. Choquet-Bruhat, \textit{Global Wave Maps on Curved Space
Times}. In: Cotsakis S., Gibbons G.W. (eds) Mathematical and Quantum Aspects
of Relativity and Cosmology. Lecture Notes in Physics, vol 537. Springer,
Berlin, Heidelberg (2000).






\end{thebibliography}
\end{document}